\documentclass[twocolumn,showpacs,superscriptaddress,nofootinbib]{revtex4-1}
\usepackage[latin9]{inputenc}
\usepackage{amsmath}
\usepackage{amssymb}
\usepackage{amsthm}
\usepackage{amsthm}

\usepackage{graphicx}
\usepackage{color}
\usepackage{enumitem}

\theoremstyle{plain}

\newtheorem{theorem}{Theorem}
\newtheorem*{theorem*}{Theorem}
\newtheorem{lemma}[theorem]{Lemma}

\makeatletter

\@ifundefined{textcolor}{}
{
 \definecolor{BLACK}{gray}{0}
 \definecolor{WHITE}{gray}{1}
 \definecolor{RED}{rgb}{1,0,0}
 \definecolor{GREEN}{rgb}{0,1,0}
 \definecolor{BLUE}{rgb}{0,0,1}
 \definecolor{CYAN}{cmyk}{1,0,0,0}
 \definecolor{MAGENTA}{cmyk}{0,1,0,0}
 \definecolor{YELLOW}{cmyk}{0,0,1,0}
}

\usepackage{color}
\usepackage{amsfonts}
\usepackage{graphics}
\usepackage{bm}
\usepackage{epsfig}\usepackage{amsthm}

\usepackage{lipsum}

\def\identity{\leavevmode\hbox{\small1\kern-3.8pt\normalsize1}}

\newcommand{\ket}[1]{\left | #1 \right\rangle}
\newcommand{\bra}[1]{\left \langle #1 \right |}

\newcommand{\proj}[1]{\ket{#1}\bra{#1}}
\renewcommand{\epsilon}{\varepsilon}

\makeatother

\begin{document}

\title{Detecting non-decomposability of time evolution via extreme gain of correlations}

\author{Tanjung Krisnanda}
\affiliation{School of Physical and Mathematical Sciences, Nanyang Technological University, 637371 Singapore, Singapore}

\author{Ray Ganardi}
\affiliation{School of Physical and Mathematical Sciences, Nanyang Technological University, 637371 Singapore, Singapore}

\author{Su-Yong Lee}
\affiliation{School of Computational Sciences, Korea Institute for Advanced Study, Hoegi-ro 85, Dongdaemun-gu, Seoul 02455, Korea}

\author{Jaewan Kim}
\affiliation{School of Computational Sciences, Korea Institute for Advanced Study, Hoegi-ro 85, Dongdaemun-gu, Seoul 02455, Korea}

\author{Tomasz Paterek}
\affiliation{School of Physical and Mathematical Sciences, Nanyang Technological University, 637371 Singapore, Singapore}
\affiliation{MajuLab, CNRS-UCA-SU-NUS-NTU International Joint Research Unit, UMI 3654 Singapore, Singapore}

\begin{abstract}
Non-commutativity is one of the most elementary non-classical features of quantum observables.
Here we propose a method to detect non-commutativity of interaction Hamiltonians of two probe objects coupled via a mediator.
If these objects are open to their local environments, our method reveals non-decomposability of temporal evolution into a sequence of interactions between each probe and the mediator. 
The Hamiltonians or Lindblad operators can remain unknown throughout the assessment, we only require knowledge of the dimension of the mediator.
Furthermore, no operations on the mediator are necessary.
Technically, under the assumption of decomposable evolution, we derive upper bounds on correlations between the probes and
then demonstrate that these bounds can be violated with correlation dynamics generated by non-commuting Hamiltonians, e.g., Jaynes-Cummings coupling.
An intuitive explanation is provided in terms of multiple exchanges of a virtual particle which lead to the excessive accumulation of correlations.
A plethora of correlation quantifiers are helpful in our method, e.g., quantum entanglement, discord, mutual information, and even classical correlation.
Finally, we discuss exemplary applications of the method in quantum information: the distribution of correlations and witnessing dimension of an object.
\end{abstract}

\maketitle

\section{Introduction}

All classical observables are functions of positions and momenta.
Since there is no fundamental limit on the precision of position and momentum measurement in classical physics,
all classical observables are, in principle, measurable simultaneously.
Quite differently, the Heisenberg uncertainty principle forbids simultaneous exact knowledge of quantum observables corresponding to position and momentum.
The underlying non-classical feature is their non-commutativity:
Any pair of non-commuting observables cannot be simultaneously measured to arbitrary precision, as first demonstrated by Robertson in his famous uncertainty relation~\cite{robertson1929}.
Other examples of non-classical phenomena with underlying non-commutativity of observables include violations of Bell inequalities~\cite{landau1987,peres-book} or, more generally, non-contextual inequalities; e.g., see~\cite{thompson2016}.
Here we describe a method to detect non-commutativity of interaction Hamiltonians, and generally non-decomposability of temporal evolution, from the dynamics of correlations.

Consider the situation depicted in Fig.~\ref{FIG_setup}, where the probe systems $A$ and $B$ do not interact directly but only via the mediator $C$; i.e., there is no Hamiltonian term $H_{AB}$.
In general, we allow all objects to be open systems and study whether the evolution operator cannot be represented by a sequence of operations between each probe and the mediator, i.e., $\Lambda_{BC}\Lambda_{AC}$ or in reverse order.
For the special case in which all systems are closed, non-decomposability implies non-commutativity of interaction Hamiltonians, i.e., $[H_{AC},H_{BC}]\ne0$.
Indeed, for commuting Hamiltonians, the unitary evolution operator is decomposable into $U_{BC}U_{AC}$, where, for example, $U_{AC} = \exp(- i t H_{AC})$ and we set $\hbar = 1$.
We show that for decomposable evolution, correlations between $A$ and $B$ are bounded. 
We also show with concrete dynamics generated by non-commuting Hamiltonians that these bounds can be violated.
The bounds derived depend solely on the dimensionality of $C$ and not on the actual form of the evolution operators.
Hence, these operators can remain unknown throughout the assessment.
This is a desired feature, as experimenters usually do not reconstruct the evolution operators via process tomography.
It also allows applications of the method to situations where the physics is not understood to the extent that reasonable Hamiltonians or Lindblad operators can be written down.
Furthermore, the assessment does not depend on the initial state of the tripartite system and does not require any operations on the mediator.
It is therefore applicable to a variety of experimental situations; Refs.~\cite{rauschenbeutel2001,sahling2015,baart2017,hamsen} provide concrete examples.

\begin{figure}[!h]
\includegraphics[width=0.4\textwidth]{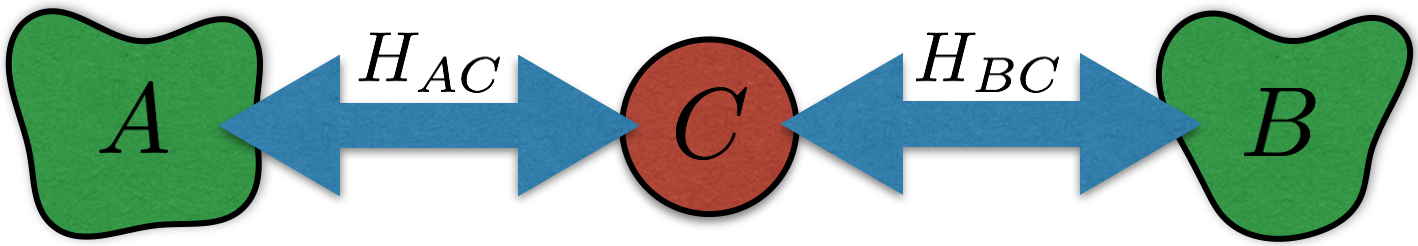}
\caption{Probe objects $A$ and $B$ individually interact with a mediator $C$, but not with each other ($A$, $B$, and $C$ could be open to their local environments). 
The coherent parts of the interactions are described by Hamiltonians $H_{AC}$ and $H_{BC}$. We show how to infer non-decomposability of the temporal evolution based on a gain of correlation between $A$ and $B$ exceeding a certain threshold, which is a function of the dimensionality of $C$ only.}
\label{FIG_setup}
\end{figure}

We begin by presenting the general bounds on the amount of correlations one can establish if the evolution is decomposable.
It is shown that these bounds are generic and hold for a plethora of correlation quantifiers.
We then calculate concrete bounds on exemplary quantifiers and show how they can be violated in a system of two fields coupled by a two-level atom.
We discuss the origin of the violation in terms of ``Trotterized" evolution, where a virtual particle is exchanged between $A$ and $B$ multiple times if the Hamiltonians do not commute but only once if they do commute.
Finally, we focus on immediate applications in quantum information and discuss the consequences of our findings for correlation distribution protocols and dimension witnesses.


\section{Results}

\subsection{General bounds}

Consider the setup illustrated in Fig. \ref{FIG_setup}.
System $C$, with finite dimension $d_C$, is mediating interactions between higher-dimensional systems $A$ and $B$.
For simplicity we take $d_A=d_B > d_C$.
We assume that there is no direct interaction between $A$ and $B$, such that the Hamiltonian of the whole tripartite system is of the form $H_{AC}+H_{BC}$ (local Hamiltonians $H_A$, $H_B$, and $H_C$ included).
Our bounds follow from a generalization of the following simple observation.
Consider, for the moment, the relative entropy of entanglement as the correlation quantifier \cite{vedral1997}.
If the evolution is decomposable, it can be written as $\Lambda_{BC} \Lambda_{AC}$, or in reverse order.
Therefore, it is as if particle $C$ interacted first with $A$ and then with $B$, a scenario similar to that in Refs.~\cite{cubitt2003,streltsov2012,chuan2012,edssexp1,edssexp2,edssexp3}.
The first interaction can generate at most $\log_2(d_C)$ ebits of entanglement, whereas the second, in the best case, can swap all this entanglement. 
In the end, particles $A$ and $B$ gain at most $\log_2(d_C)$ ebits.
The bound is indeed independent of the form of interactions. 
Furthermore, it is intuitively clear, as this is just the ``quantum capacity'' of the mediator. 

Now let us consider correlation quantifiers obtained in the so-called ``distance" approach~\cite{vedral1997, modi2010}.
The idea is to quantify correlation $Q_{X:Y}$ in a state $\rho_{XY}$ as the shortest distance $D(\rho_{XY},\sigma_{XY})$ from $\rho_{XY}$ to a set of states $\sigma_{XY} \in \mathcal{S}$ without the desired correlation property,
i.e., $Q_{X:Y} \equiv \inf_{\sigma_{XY}\in \mathcal{S}} D(\rho_{XY},\sigma_{XY})$.
For example, the relative entropy of entanglement is given by the relative entropy of a state to the set of disentangled states~\cite{vedral1997}.
It turns out that most such quantifiers are useful for the task introduced here.
The conditions we require are that (i) $\mathcal{S}$ is closed under local operations $\Lambda_Y$ on $Y$,
(ii) $D(\Lambda[\rho],\Lambda[\sigma]) \leq D(\rho, \sigma)$ (monotonicity), and (iii) $D(\rho_0, \rho_1) \leq D(\rho_0, \rho_2) + D(\rho_2, \rho_1)$ (triangle inequality).
They are sufficient to prove the following theorem.
\begin{theorem}\label{TH_cons}
Suppose a correlation $Q_{X:Y}$ satisfies properties {\rm(i)--(iii)} listed above.
If the evolution operator $\Lambda_{ABC}$ is decomposable into $\Lambda_{BC}\Lambda_{AC}$, then
\begin{equation}
	Q_{A:B} (t) \leq I_{AC:B}(0) + \sup_{\ket{\psi}} Q_{A:C},
	\label{TH_IQ_BOUND}
\end{equation}
where $I_{AC:B}(0) = \inf_{\sigma_{AC} \otimes \sigma_B} D(\rho, \sigma_{AC} \otimes \sigma_B)$, $\rho$ is the initial tripartite state, and the supremum of $Q_{A:C}$ is taken over pure states of $AC$.
\end{theorem}
\begin{proof}
The proof is given in Appendix A.
\end{proof}
Note that although the relative entropy does not satisfy (iii) it still follows Theorem \ref{TH_cons} (cf. Lemma \ref{LM_re} in Appendix A). 
Correlations between probe $A$ and probe $B$ are therefore bounded by the maximal achievable correlation with the mediator, $\sup_{\ket{\psi}} Q_{A:C}$.
The additional term $I_{AC:B}(0)$ reduces to the usual mutual information if $D(\rho_{XY},\sigma_{XY})$ is the relative entropy distance~\cite{modi2010} 
and characterizes the amount of total initial correlations between one of the probes and the rest of the system.
Note that the bound is independent of time.
This can be seen as a result of the effective description of such dynamics given by $\Lambda_{BC}\Lambda_{AC}$.
The particle $C$ is exchanged between $A$ and $B$ only once, independently of the duration of the dynamics. 

In a typical experimental situation the initial state can be prepared as completely uncorrelated $\rho = \rho_A \otimes \rho_B \otimes \rho_C$,
in which case Theorem~\ref{TH_cons} simplifies and the bound is given solely in terms of the ``correlation capacity'' of the mediator:
\begin{equation}
Q_{A:B} (t) \le \sup_{\ket{\psi}} Q_{A:C}.
\label{EQ_CC}
\end{equation}
Clearly, the same bound holds for initial states of the form $\rho = \rho_{AC} \otimes \rho_B$.
In Appendix B we show that, with this initial state, Eq.~(\ref{EQ_CC}) holds for any correlation quantifier that is monotonic under local operations $\Lambda_{BC}$, 
not necessarily based on the distance approach, e.g., any entanglement monotone.

For initial states that are close to $\rho = \rho_{AC} \otimes \rho_B$ one can utilize the continuity of the von Neumann entropy \cite{fannes1973continuity} and see that $I_{AC:B}(0)$ in Eq. (\ref{TH_IQ_BOUND}) is indeed small. We can also ensure that the initial state is of the form $\rho = \rho_{AC} \otimes \rho_B$  by performing a correlation breaking channel on $B$ first.
One example of such a channel is a measurement in the computational basis followed by a measurement in some complementary (say Fourier) basis.
This implements the correlation breaking channel $(\identity_{AC} \otimes \Lambda_B) (\rho_{ABC}) = \rho_{AC} \otimes \frac{\identity}{d_B}$~\cite{footnote}.
In this way, our method does not require any knowledge of the initial state and any operations on the mediator, similar in spirit to the detection of quantum discord of inaccessible objects in Ref.~\cite{krisnanda2017}.
We now move to concrete correlation quantifiers and their correlation capacities.


\subsection{Exemplary measures and bounds}

We provide four correlation quantifiers which capture different types of correlations between quantum particles.
All of them are shown to be useful in detecting non-decomposability.

Mutual information is a measure of total correlations~\cite{groisman2005} and is defined as $I_{X:Y} = S_X + S_Y - S_{XY}$, where, e.g., $S_X$ is the von Neumann entropy of subsystem $X$.
It can also be seen as a distance-based measure with the relative entropy as the distance and a set of product states $\sigma_X \otimes \sigma_Y$ as $\mathcal{S}$~\cite{modi2010}.
The supremum in Eq. (\ref{EQ_CC}) is attained by the state (recall that $d_A > d_C$),
\begin{equation}
|\Psi \rangle = \frac{1}{\sqrt{d_C}} \sum_{j = 1}^{d_C} |a_j \rangle |c_j \rangle,
\label{EQ_MAX_ENT}
\end{equation}
where $| a_j \rangle$ and $| c_j \rangle$ form orthonormal bases. 
One finds $\sup_{\ket{\psi}} I_{A:C} = 2 \log_2(d_C)$.

An interesting quantifier in the context of non-classicality detection is the classical correlation in a quantum state.
It is defined as mutual information of the classical state obtained by performing the best local von Neumann measurements on the original state $\rho$~\cite{terhal2002},
i.e., $C_{X:Y} = \sup_{\Pi_X \otimes \Pi_Y} I_{X:Y}(\Pi_X \otimes \Pi_Y(\rho) )$, where $\Pi_X \otimes \Pi_Y(\rho) = \sum_{xy} \proj{xy} \rho \proj{xy}$, and $\ket{x}$, $\ket{y}$ form orthonormal bases.
The supremum of mutual information over classical states of $AC$ is $\log_2(d_C)$.

Quantum discord is a form of purely quantum correlations that contain quantum entanglement.
It can be phrased as a distance-based measure. 
In particular, we consider the relative entropy of discord~\cite{modi2010}, also known as the one-way deficit~\cite{horodecki2005}.
It is an asymmetric quantity defined as $\Delta_{X|Y} = \inf_{\Pi_Y} S(\Pi_Y(\rho)) - S(\rho)$, where $\Pi_Y$ is a von Neumann measurement conducted on subsystem $Y$.
The relative entropy of discord is maximized by the state (\ref{EQ_MAX_ENT}), for which we have $\sup_{\ket{\psi}} \Delta_{A|C} = \log_2(d_C)$.

Our last example is negativity, a computable entanglement monotone~\cite{vidal2002}.
For a bipartite system negativity is defined as $N_{X:Y} = (||\rho^{T_X}||_1-1)/2$,
where $||.||_1$ denotes the trace norm and $\rho^{T_X}$ is a matrix obtained by partial transposition of $\rho$ with respect to $X$.
Negativity is maximized by the state (\ref{EQ_MAX_ENT}), and the supremum reads $\sup_{\ket{\psi}} N_{A:C} = (d_C-1)/2$.

Clearly, many other correlation quantifiers are suitable for our detection method because the assumptions behind Eqs.~(\ref{TH_IQ_BOUND}) and (\ref{EQ_CC}) are not demanding.
In fact, one may wonder which correlations do not qualify for our method.
A concrete example is the geometric quantum discord based on $p$-Schatten norms with $p > 1$, as it may increase under local operations on $BC$~\cite{piani2012,paula2013}.


\subsection{Violations}

We now demonstrate, with concrete dynamics generated by non-commuting Hamiltonians, that the bounds derived can be violated.
We next discuss the origin of this violation.

Consider a two-level atom $C$, i.e., $d_C=2$, mediating interactions between two cavity fields $A$ and $B$. 
A similar scenario has been considered and implemented, for example, in Refs. \cite{rauschenbeutel2001,messina2002,browne2003,hamsen}.
The interaction between the atom and each cavity field is taken to follow the Jaynes-Cummings model,
\begin{equation}
H = g (\hat a\hat \sigma_+ +\hat a^{\dagger}\hat \sigma_-)+g (\hat b\hat \sigma_+ +\hat b^{\dagger}\hat \sigma_-),
\label{EQ_JC}
\end{equation}
where $\hat a$ ($\hat b$) is the annihilation operator of field $A$ ($B$), while $\hat \sigma_+$ ($\hat \sigma_-$) is the raising (lowering) operator of the two-level atom.
For simplicity, we have assumed that the interaction strengths between the two-level atom and the fields are the same.
Note that $H$ is of the form $H_{AC}+H_{BC}$ with non-commuting components.

The resulting correlation dynamics are plotted in Fig. \ref{FIG_dynamics}. 
Mutual information and negativity were calculated directly, whereas for the classical correlation and the relative entropy of discord, we provide the lower bounds $\tilde C_{A:B}$ and $-S_{A|B}$, respectively.
$\tilde C_{A:B}$ is calculated as the mutual information of the state resulting from projective local measurements in the Fock basis (no optimization over measurements performed).
The negative conditional entropy $-S_{A|B}$ is a lower bound on the distillable entanglement \cite{devetak2005distillation}, which in turn is a lower bound on the relative entropy of entanglement $E_{A:B}$ \cite{horodecki2000limits}. Therefore, we note the chain of inequalities $-S_{A|B} \le E_{A:B} \le \Delta_{A|B} \le I_{A:B}$, where the last two inequalities follow from \cite{modi2010}.
Already these lower bounds can beat the limit set by decomposable evolution, and therefore, all mentioned correlations can detect non-decomposability of the evolution. 
Since we consider closed systems, this infers non-commutativity of the Jaynes-Cummings couplings.
We also note another non-classical feature of the studied dynamics: since Fig. \ref{FIG_dynamics} shows entanglement gain, according to Ref.~\cite{krisnanda2017} there must be quantum discord $D_{AB|C}$ during the evolution.

\begin{figure}[!h]
\includegraphics[scale=0.28]{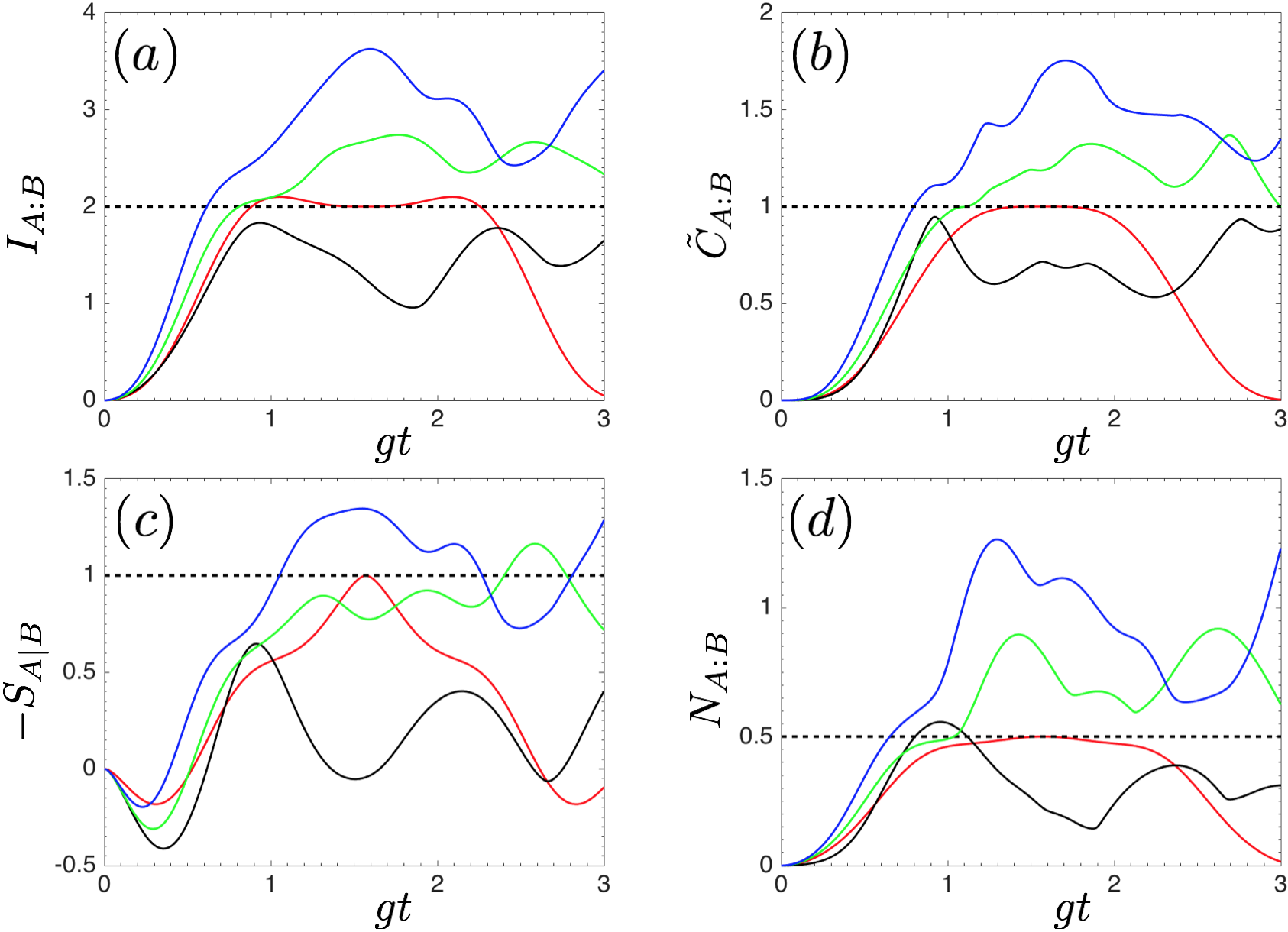}
\caption{Correlation dynamics with the Jaynes-Cummings model (solid curves) and the corresponding bounds for decomposable evolution (dashed lines).
(a) Mutual information, (b) lower bound on the classical correlation (see the main text), (c) lower bound on the relative entropy of discord, and (d) negativity.
In all cases, time is rescaled with the interaction strength $g$ and the initial state of $ABC$ is varied: $\ket{110}$ (red curves), $\ket{101}$ (black curves), $\ket{210}$ (green curves), and $\ket{220}$ (blue curves).}
\label{FIG_dynamics}
\end{figure}

It is apparent that the detection is easier (faster and with more pronounced violation) with a higher number of photons in the initial states of the cavity fields.
We offer an intuitive explanation.
Consider, for example, $| m n 0 \rangle$ as the initial state of $ABC$.
By defining $\hat \xi = (\hat a +\hat b )/\sqrt{2}$, the Hamiltonian of Eq. (\ref{EQ_JC}) becomes $\sqrt{2} g(\hat \xi \hat \sigma_+ + \hat \xi^{\dagger}\hat \sigma_-)$ and it is straightforward to obtain the unitary evolution~\cite{scully-book}.
One finds that the quantum state of the fields oscillates incoherently between $\sum^{m+n}_{j=0}c_j(t) |j\rangle_A |m+n - j\rangle_B$ and $\sum^{m+n-1}_{j=0}d_j(t) |j\rangle_A |m+n-1-j\rangle_B$.
Both of these states are superpositions of essentially $m+n$ bi-orthogonal terms giving rise to high entanglement and, therefore, also other forms of correlations.

Figure~\ref{FIG_dynamics} illustrates that different correlation quantifiers have different detection capabilities and it is not clear at this stage whether there is a universal measure with which non-commutativity is detected, e.g., the fastest.
For most initial states we studied mutual information detected non-commutativity the most rapidly, but there are exceptions, as shown by the black curve corresponding to the initial state $| 1 0 1\rangle$.
With this initial state the mutual information never violates its bound, but the negativity does.


\section{Discussion}

Let us present the origin of the violation just observed.
Since the total Hamiltonian is of the form $H_{AC} + H_{BC}$, the Suzuki-Trotter expansion of the resulting evolution is particularly illuminating,
\begin{equation}\label{EQ_trotterexp}
e^{i t (H_{AC} + H_{BC})} = \lim_{n \to \infty} \left( e^{i\Delta t H_{BC}} e^{i\Delta t H_{AC}} \right)^n,
\end{equation}
where $\Delta t = t / n$.
If Hamiltonians do not commute, it is necessary to think about Eq. (\ref{EQ_trotterexp}) as $n$ sequences of pairwise interactions of $C$ with $A$ followed by $C$ with $B$, each for a time $\Delta t$.
Each pair of interactions can only increase correlations up to the correlation capacity of the mediator, but their multiple use allows the accumulation of correlations beyond what is possible with commuting Hamiltonians.
Recall that, in the latter case, we deal with only one exchange of system $C$, independently of the duration of dynamics.
We stress that Trotterization is just a mathematical tool and in the laboratory system $C$ is continuously coupled to $A$ and $B$.
It is rather as if a virtual particle $C$ were transmitted multiple times between $A$ and $B$, interacting with each of them for a time $\Delta t$.

Our results imply that the non-commutativity (non-decomposability in general) is a desired feature of interactions in the task of correlation distribution, which is important for quantum information applications.
As a contrasting physical illustration, we consider the strong dipole-dipole interactions in our field-atom-field example.
The Hamiltonian reads
\begin{equation}
H^{\prime}=g (\hat a+\hat a^{\dagger}) (\hat \sigma_+ + \hat \sigma_-)+g (\hat b+\hat b^{\dagger})(\hat \sigma_+ + \hat \sigma_-),
\label{EQ_DD}
\end{equation}
with commuting components, i.e., $[H_{AC},H_{BC}]=0$.
One can verify that the results of this model are in agreement with all the bounds we derived.
Furthermore, we prove in Appendix C that, with this coupling, the state of $AB$ at time $t$ is effectively given by a two-qubit separable state. This makes $N_{A:B}(t)=0$ and $I_{A:B}(t)\le 1$. 
Note the counter-intuitive result that strong interactions produce bounded correlations between the probes, while weak interactions (Jaynes-Cummings coupling) can increase the correlations above the bounds. 

We also note an application of our bounds to estimate the dimension of the mediator; see, e.g., Refs.~\cite{dimwit1,dimwit2,dimwit3} for other dimension witnesses.
For decomposable evolution (including discrete sequential operators considered in Refs.~\cite{cubitt2003,streltsov2012,chuan2012,edssexp1,edssexp2,edssexp3}), the amount of correlation between the probes is bounded by the correlation capacity $\sup_{\ket{\psi}} Q_{A:C}$, which is a function of $d_C$.
If one observes a $Q_{A:B}(t)$ value that is larger than the correlation capacity of a certain $d_C$, then the dimension of the mediator must be larger than $d_C$.

Finally, we wish to discuss a scenario where the three systems are open to their own \emph{local} environments, as realized, e.g., in~\cite{hamsen}.
We take the evolution following the master equation in Lindblad form,
\begin{eqnarray}
\dot \rho & = & -i[H_{AC}+H_{BC},\rho]+\sum_{X=A,B,C}L_X\rho, \label{EQ_open}\\
L_X\rho & = & \sum_k Q^X_k\rho Q^{X\dag}_k-\frac{1}{2}\{Q^{X\dag}_kQ^X_k,\rho\}, \nonumber
\end{eqnarray}
where the last term in (\ref{EQ_open}) is the incoherent part of the evolution and $L_X$ describes interactions of system $X$ with its local environment, i.e. the operators $Q^X_k$ act on system $X$ only.
We denote $\mathcal{L}_{AC}=-i[H_{AC},\cdot]+L_A+L_C$ and $\mathcal{L}_{BC}=-i[H_{BC},\cdot]+L_B$. 
One readily verifies that if $[H_{AC},H_{BC}]=0$ and $[L_C,H_{BC}]=0$, we have commuting Lindblad operators $[\mathcal{L}_{AC},\mathcal{L}_{BC}]=0$. 
Note that, if one includes $L_C$ in $\mathcal{L}_{BC}$ instead, the second condition for commuting Lindblad operators now reads $[L_C,H_{AC}]=0$.
For commuting Lindbladians, the corresponding evolution decomposes as $\Lambda_{BC} \Lambda_{AC}$, or in reverse order. 
Therefore, our bounds apply accordingly. 
Their violation implies that either the Hamiltonians do not commute or the operators describing dissipative channels on $C$ do not commute with $H_{AC}$ and $H_{BC}$.
In particular, if $C$ is kept isolated so that its noise can be ignored, the violation of our bounds is solely the result of the non-commutativity of the Hamiltonians.


\section{Conclusions}

We linked non-commutativity of interaction Hamiltonians (non-decomposability of time evolution in general) to the amount of correlations that can be created in the associated dynamics.
This led us to a method for detection of non-decomposability of evolution in a scenario where subsystem $C$ mediates interactions between $A$ and $B$ (all these objects can interact with their local environments).
The method requires no explicit form of the evolution operators or knowledge of the initial state of the tripartite system.
Non-decomposability is detected by observing violation of certain bounds on $AB$ correlations, as measured by most correlation quantifiers.
Furthermore, no operation on $C$ is necessary at any time, which makes this strategy experimentally friendly.
In particular, in addition to avoiding characterization of the interactions, the physics of $C$ can remain largely unknown---only its dimension should be identified.\\

\section*{Acknowledgments}

We thank Alexander Streltsov and Kavan Modi for insightful discussions, and Matthew Lake for comments on the manuscript. 
S.-Y. L. would like to thank Changsuk Noh for useful comments.
T. K. and T. P. thank Wies{\l}aw Laskowski for hospitality at the University of Gda{\'n}sk.
This work is supported by Singapore Ministry of Education Academic Research Fund Tier 2 Project No. MOE2015-T2-2-034. 

\appendix

\section{$\mbox{Proof of Theorem 1}$}
\label{APP_TH_BOUND}

For completeness let us begin with a useful lemma.

\setcounter{theorem}{0}
\begin{lemma}\label{LM_pre}
For a measure of correlations $Q_{X:Y}$ between party $X$ and party $Y$ that is non-increasing under local operations on $Y$, the following property holds: $Q_{X:Y}$ is invariant under tracing-out of uncorrelated systems on the side of $Y$.
\end{lemma}
\begin{proof}
Since the correlation measure is non-increasing under local operations on $Y$, tracing out an uncorrelated system on the side of $Y$ can only decrease the correlation.
However, if the correlation is strictly decreasing, then there is a reverse process, i.e., attaching the uncorrelated system back and, therefore, increasing the correlation.
Hence, the correlation $Q_{X:Y}$ has to be invariant under tracing-out of uncorrelated systems on $Y$.
In fact, this is true for all reversible operations.
\end{proof}

Our main theorem is proven as follows.

\setcounter{theorem}{0}
\begin{theorem}\label{TH_cons}
Consider a correlation measure $Q_{X:Y} \equiv \inf_{\sigma_{XY}\in \mathcal{S}} D(\rho_{XY},\sigma_{XY})$ satisfying the following properties:
\begin{enumerate}
\item[{\rm (i)}] $\mathcal{S}$ is closed under local operations $\Lambda_Y$ on $Y$;
\item[{\rm (ii)}] $D(\Lambda[\rho],\Lambda[\sigma]) \leq D(\rho, \sigma)$; and
\item[{\rm (iii)}] $D(\rho_0, \rho_1) \leq D(\rho_0, \rho_2) + D(\rho_2, \rho_1)$.
\end{enumerate}
If the evolution operator $\Lambda_{ABC}$ is decomposable into $\Lambda_{BC}\Lambda_{AC}$, then
\begin{equation}
	Q_{A:B} (t) \leq I_{AC:B}(0) + \sup_{\ket{\psi}} Q_{A:C},
\end{equation}
where $I_{AC:B}(0) = \inf_{\sigma_{AC} \otimes \sigma_B} D(\rho, \sigma_{AC} \otimes \sigma_B)$, $\rho$ is the initial tripartite state, and the supremum of $Q_{A:C}$ is taken over pure states of $AC$.
\end{theorem}
\begin{proof}
Properties (i) and (ii), and the definition of $Q_{X:Y}$ as the shortest distance, imply that $Q_{X:Y}$ is nonincreasing under local operations on $Y$. 
Accordingly, the property proven in Lemma~\ref{LM_pre} applies.
We have
\begin{eqnarray}
	Q_{A:B} (t)
	&\leq& Q_{A:BC} \left(\Lambda_{BC} \Lambda_{AC} [\rho] \right) \label{APP_TH_S1} \\
	&\le& Q_{A:BC} \left( \Lambda_{AC} [\rho] \right) \label{APP_TH_S2} \\
	& \le & D \left(\Lambda_{AC} [\rho], \mu \right) \label{APP_TH_S3} \\
	& \leq & D \left( \Lambda_{AC} [\rho], \Lambda_{AC} [\sigma^0_{AC}]  \otimes \sigma^0_B \right) \nonumber \\
	& + & D\left( \Lambda_{AC} [\sigma^0_{AC}]  \otimes \sigma^0_{B}, \mu \right) \label{APP_TH_TRIAN}\\
	& \le &  D(\rho, \sigma^0_{AC} \otimes \sigma^0_B)	\nonumber \\
	&+ & D\left( \Lambda_{AC} [\sigma^0_{AC}]  \otimes \sigma^0_{B}, \mu \right) \label{APP_TH_S4} \\
	&=& I_{AC:B}(0) + Q_{A:BC} (\Lambda_{AC} [\sigma^0_{AC}]  \otimes \sigma^0_{B} ) \label{APP_TH_S5} \\
	&=& I_{AC:B}(0) + Q_{A:C} (\Lambda_{AC} [\sigma^0_{AC}] ) \label{APP_TH_S6} \\
	&\leq& I_{AC:B}(0) +\sup_{\ket{\psi}} Q_{A:C}, \label{APP_TH_S7}
	\end{eqnarray}
where the steps are justified as follows.
In line (\ref{APP_TH_S1}) we have used the fact that $Q_{X:Y}$ is nonincreasing under local operations on $Y$ (tracing out $C$).
Line (\ref{APP_TH_S2}) follows, as $Q_{A:BC}$ is nonincreasing under local operation $\Lambda_{BC}$.
The next line, (\ref{APP_TH_S3}), utilizes the definition of $Q_{A:BC}$ as the shortest distance to the set of states $\mu \in \mathcal{S}_{A:BC}$.
The inequality of (\ref{APP_TH_TRIAN}) follows from the triangle inequality (iii).
Note that the first distance in (\ref{APP_TH_TRIAN}) does not depend on $\mu$ and at this point one can choose any $\sigma^0_{AC}$ and $\sigma^0_B$.
The inequality (\ref{APP_TH_S4}) invokes property (ii).
In (\ref{APP_TH_S5}), we have chosen $\sigma^0_{AC} \otimes \sigma^0_B$ as the closest product state to $\rho$ and $\mu$ as a state in $\mathcal{S}_{A:BC}$ closest to $\Lambda_{AC} [\sigma^0_{AC}]  \otimes \sigma^0_{B}$.
Line (\ref{APP_TH_S6}) uses the invariance of $Q_{A:BC}$ under tracing-out of the uncorrelated system $\sigma^0_B$.
For the final inequality, we note that a correlation measure that is nonincreasing under local operations on at least one side must be maximal on pure states~\cite{streltsov2012general}.
\end{proof}

\setcounter{theorem}{1}
\begin{lemma}\label{LM_re}
The conclusion in Theorem \ref{TH_cons} still follows for the relative entropy as a distance measure.
\end{lemma}
\begin{proof}
Let us begin with an identity, 
\begin{eqnarray}
S(\rho||\sigma_X\otimes \sigma_Y)&=&\mbox{tr}(\rho \log{\rho}-\rho\log{\sigma_X \otimes \sigma_Y})  \nonumber \\
&=& \mbox{tr}(\rho \log{\rho}-\rho \log{\rho_X \otimes \rho_Y})  \nonumber \\
&&+\mbox{tr}(\rho \log{\rho_X \otimes \rho_Y} - \rho \log{\sigma_X \otimes \sigma_Y}) \nonumber \\
&=&S(\rho||\rho_X\otimes \rho_Y)+S(\rho_X||\sigma_X) \nonumber \\
&&+S(\rho_Y||\sigma_Y), \label{GG0}
\end{eqnarray}
where $\rho_X$ and $\rho_Y$ are the marginals of $\rho$ and we have used, for example, relation $\mbox{tr}(\rho \log{\sigma_X \otimes \sigma_Y}) = \mbox{tr}(\rho_X \log{\sigma_X}) + \mbox{tr}(\rho_Y \log{\sigma_Y})$.

Although relative entropy satisfies (ii) \cite{uhlmann1977relative}, it is well known not to follow (iii). Therefore, starting from (\ref{APP_TH_S2}), we have
\begin{eqnarray}
&&Q_{A:BC} \left( \Lambda_{AC} [\rho] \right) \nonumber \\
&=&\inf_{\mu \in \mathcal{S}_{A:BC}} S \left(\Lambda_{AC}[\rho] || \mu \right) \label{GG01}\\
&\le& S(\Lambda_{AC}[\rho] || \mu_{AC}\otimes \mu_B) \label{GG1}\\
&=&  S(\Lambda_{AC}[\rho] || \rho_{AC}^{\prime}\otimes \rho_B^{\prime}) \nonumber \\
&&+S(\rho_{AC}^{\prime} || \mu_{AC})+S(\rho_{B}^{\prime} || \mu_{B}) \label{GG2} \\
&=&I_{AC:B}(\Lambda_{AC}[\rho] )+Q_{A:C}(\rho_{AC}^{\prime}) \label{GG3} \\
&\le&I_{AC:B}(0)+\sup_{\ket{\psi}} Q_{A:C},
\end{eqnarray}
where $\rho_{AC}^{\prime}$ and $\rho_B^{\prime}$ are marginals of $\Lambda_{AC}[\rho]$. 
The steps above are justified as follows. 
Line (\ref{GG1}) follows for any state of the form $\mu_{AC}\otimes \mu_B\in \mathcal{S}_{A:BC}$.
We have used identity (\ref{GG0}) in line (\ref{GG2}). The equality (\ref{GG3}) uses the definition of mutual information as the relative entropy from a state to its marginals \cite{modi2010}. We have also chosen $\mu_{AC}$ as a state in $\mathcal{S}_{A:C}$ closest to $\rho_{AC}^{\prime}$ and $\mu_B=\rho_B^{\prime}$. The last line follows as mutual information is non-increasing under local operation $\Lambda_{AC}$ and the correlation $Q_{A:C}$ achieves the supremum on pure states. 
\end{proof}

\section{$\mbox{Proof of Eq. (2) for correlations only}$\\  $\mbox{monotonic under local operations}\:\Lambda_{BC}$}
\label{APP_NOND}

\setcounter{theorem}{1}
\begin{theorem}\label{TH_main}
	Suppose the initial state has the form $\rho = \rho_{AC} \otimes \rho_B$.
	If the evolution operator is decomposable into $\Lambda_{BC}\Lambda_{AC}$, then 
	\begin{equation}
	Q_{A:B}(t)\le \sup_{\ket{\psi}} \:Q_{A:C}
	\end{equation}
	 for all correlation measures, $Q$, non-increasing under local operations $\Lambda_{BC}$.
\end{theorem}
\begin{proof}
For initial states of the form $\rho_{AC}\otimes \rho_B$ we have the following chain of arguments
\begin{eqnarray}
Q_{A:B}(t) & \le & Q_{A:BC}(t)\le Q_{A:BC}(\Lambda_{AC}[\rho]) \nonumber \\
&=& Q_{A:C}(\Lambda_{AC}[\rho]) \le \sup_{\ket{\psi}}  Q_{A:C},
\end{eqnarray}
where the steps are justified as follows.
Since the action of tracing out (the, in general, correlated) system $C$ is a local operation on $BC$, we obtain the first inequality.
The second inequality follows as the correlation is non-increasing under $\Lambda_{BC}$.
As we start with the initial state $\rho_{AC}\otimes \rho_B$ and $\Lambda_{AC}$ does not act on $B$, system $B$ stays uncorrelated in $\Lambda_{AC}[\rho]$.
Using Lemma \ref{LM_pre}, we have the equality.
Finally, the correlation $Q_{A:C}$ is again maximal on pure states.
\end{proof}
\vspace{0.5cm}

\section{$\mbox{Proof of separability via dipole-dipole}$\\ $\mbox{coupling for particular initial states}$}
\label{APP_N}

Let us define $\hat{\xi} = (\hat a +\hat b )/\sqrt{2}$. 
The dipole-dipole Hamiltonian, Eq. (6), is reformulated as $H^{\prime}=\sqrt{2}g(\hat \xi +\hat \xi^{\dagger})\hat \sigma_x$, where $\hat \sigma_x=\hat \sigma_+ +\hat \sigma_-$ and $[\hat \xi, \hat \xi^{\dagger}]=\openone$. 
The unitary evolution operator is given by
\begin{eqnarray}\label{EQ_hddunitary}
\hat U_t&=&e^{-iH^{\prime}t} \\
&=&\frac{1}{2} [(\openone-\hat \sigma_x)e^{i\sqrt{2}gt(\hat \xi + \hat \xi^{\dagger})}+(\openone+\hat \sigma_x)e^{-i\sqrt{2}gt(\hat \xi + \hat \xi^{\dagger})}]\nonumber \\
&=&\frac{1}{2} [(\openone-\hat \sigma_x)\hat D_a(\alpha)\hat D_b(\alpha) +(\openone+\hat \sigma_x)\hat D_a(-\alpha)\hat D_b(-\alpha)], \nonumber
\end{eqnarray}
where $\alpha=igt$ and, e.g., $\hat D_a (\alpha)=\exp(\alpha\hat a^{\dagger}-\alpha^{\ast}\hat a)$. 
Given an initial state $|mn0\rangle$, the state at time $t$ reads
\begin{eqnarray}
| \psi_t \rangle & = & \frac{1}{4} [ (d^{(mn)}_{++}|D^{(m)}_+,D^{(n)}_+\rangle + d^{(mn)}_{--}|D^{(m)}_-,D^{(n)}_-\rangle )|0\rangle \nonumber \\
& - & (d^{(mn)}_{+-}|D^{(m)}_+,D^{(n)}_-\rangle+d^{(mn)}_{-+}|D^{(m)}_-,D^{(n)}_+\rangle)|1\rangle], \nonumber
\end{eqnarray}
where
\begin{eqnarray}
d^{(mn)}_{\pm\pm} & = & 2\sqrt{[1\pm e^{-2|\alpha|^2} L_m(4|\alpha|^2)][1\pm e^{-2|\alpha|^2} L_n(4|\alpha|^2)]},\nonumber \\
|D^{(n)}_{\pm}\rangle & = & \frac{1}{\sqrt{d^{(nn)}_{\pm\pm}}}[\hat{D}(\alpha)\pm\hat{D}(-\alpha)]|n\rangle . \nonumber
\end{eqnarray}
Note that $\langle D^{(n)}_+|D^{(n)}_-\rangle=0$ and $\langle D^{(n)}_{\pm}|D^{(n)}_{\pm}\rangle=1$. $L_n(|\alpha|^2)$ is the Laguerre polynomial, which comes from the relation
$\langle n|\hat{D}(\alpha)|n\rangle=e^{-|\alpha|^2/2}L_n(|\alpha|^2)$.
After tracing-out of the atomic mode $C$, the state of the fields is effectively given by a two-qubit state,
\begin{eqnarray}
\frac{1}{16}
\begin{pmatrix}
(d^{(mn)}_{++})^2 & 0& 0& d^{(mn)}_{++}d^{(mn)}_{--} \\
0 & (d^{(mn)}_{+-})^2 & d^{(mn)}_{++}d^{(mn)}_{--} & 0 \\
0 & d^{(mn)}_{++}d^{(mn)}_{--} & (d^{(mn)}_{-+})^2 & 0 \\
d^{(mn)}_{++}d^{(mn)}_{--} & 0 & 0 & (d^{(mn)}_{--})^2 \nonumber
\end{pmatrix},\\
\end{eqnarray}
which is positive under partial transposition and, hence, separable \cite{peres1996,horodecki1996m}. The same result follows for initial state $\ket{mn1}$.


\begin{thebibliography}{99}

\bibitem{robertson1929}
H.~P. Robertson,
\newblock The uncertainty principle,
\newblock {Phys. Rev.} {\bf 34}, 163 (1929).

\bibitem{landau1987}
L.~J. Landau,
\newblock On the violation of Bell's inequality in quantum theory,
\newblock {Phys. Lett. A} {\bf 120}, 54 (1987).

\bibitem{peres-book}
A.~Peres,
\newblock {\em Quantum Theory: Concepts and Methods}
\newblock (Kluwer Academic, Dordrecht Netherlands, 2002).

\bibitem{thompson2016}
J.~Thompson, P.~Kurzy\'nski, S.-Y. Lee, A.~Soeda, and D.~Kaszlikowski,
\newblock Recent advances in contextuality tests,
\newblock {Open Syst. Info. Dynam.} {\bf 23}, 1650009 (2016).

\bibitem{rauschenbeutel2001}
A.~Rauschenbeutel, P.~Bertet, S.~Osnaghi, G.~Nogues, M.~Brune, J.~M. Raimond,
  and S.~Haroche,
\newblock Controlled entanglement of two field modes in a cavity quantum
  electrodynamics experiment,
\newblock {Phys. Rev. A} {\bf 64}, 050301 (2001).

\bibitem{sahling2015}
S.~Sahling, G.~Remeny, C.~Paulsen, P.~Monceau, V.~Saligrama, C.~Marin,
  A.~Revcolevschi, L.~P. Regnault, S.~Raymond, and J.~E. Lorenzo,
\newblock Experimental realization of long-distance entanglement between spins
  in antiferromagnetic quantum spin chains,
\newblock {Nat. Phys.} {\bf11}, 255 (2015).

\bibitem{baart2017}
T.~A. Baart, T.~Fujita, C.~Reichl, W.~Wegscheider, and L.~M.~K. Vandersypen,
\newblock Coherent spin-exchange via a quantum mediator,
\newblock {Nat. Nanotechnol.} {\bf12}, 26 (2017).

\bibitem{hamsen}
C. Hamsen, K. N. Tolazzi, T. Wilk, and G. Rempe,
\newblock Strong coupling between photons of two light fields mediated by one atom,
\newblock {Nat. Phys.} {\bf14}, 885 (2018).

\bibitem{vedral1997}
V.~Vedral, M.~B. Plenio, M.~A. Rippin, and P.~L. Knight,
\newblock Quantifying entanglement,
\newblock {Phys. Rev. Lett.} {\bf78}, 2275 (1997).

\bibitem{cubitt2003}
T.~S. Cubitt, F.~Verstraete, W.~D\"ur, and J.~I. Cirac,
\newblock Separable states can be used to distribute entanglement,
\newblock {Phys. Rev. Lett.} {\bf91}, 037902 (2003).

\bibitem{streltsov2012}
A.~Streltsov, H.~Kampermann, and D.~Bru\ss,
\newblock Quantum cost for sending entanglement,
\newblock {Phys. Rev. Lett.} {\bf108}, 250501 (2012).

\bibitem{chuan2012}
T.~K. Chuan, J.~Maillard, K.~Modi, T.~Paterek, M.~Paternostro, and M.~Piani,
\newblock Quantum discord bounds the amount of distributed entanglement,
\newblock {Phys. Rev. Lett.} {\bf109}, 070501 (2012).

\bibitem{edssexp1}
A. Fedrizzi, M. Zuppardo, G. G. Gillett, M. A. Broome, M. P. Almeida, M. Paternostro, A. G. White, and T. Paterek,
\newblock Experimental distribution of entanglement with separable carriers,
\newblock {Phys. Rev. Lett.} {\bf111}, 230504 (2013).

\bibitem{edssexp2}
C. E. Vollmer, D. Schulze, T. Eberle, V. H{\"a}ndchen, J. Fiur{\'a}{\v s}ek, and R. Schnabel,
\newblock Experimental entanglement distribution by separable states,
\newblock {Phys. Rev. Lett.} {\bf111}, 230505 (2013).

\bibitem{edssexp3}
C. Peuntinger, V. Chille, L. Mi{\v s}ta, N. Korolkova, M. F{\"o}rtsch, J. Korger, C. Marquardt, and G. Leuchs,
\newblock Distributing entanglement with separable states,
\newblock {Phys. Rev. Lett.} {\bf111}, 230506 (2013).

\bibitem{modi2010}
K.~Modi, T.~Paterek, W.~Son, V.~Vedral, and M.~Williamson,
\newblock Unified view of quantum and classical correlations,
\newblock {Phys. Rev. Lett.} {\bf104}, 080501 (2010).





\bibitem{fannes1973continuity}
M.~Fannes,
\newblock A continuity property of the entropy density for spin lattice
  systems,
\newblock {Commun. Math. Phys.} {\bf31}, 291 (1973).

\bibitem{footnote}
In this context one might ask whether there exists a channel such that
  $(\identity_{AC} \otimes \Lambda_{B}) (\rho_{ABC}) = \rho_{AC} \otimes
  \rho_{B}$, where $\rho_{AC}$ and $\rho_{B}$ are reduced density matrices.
  Such a channel does not exist, e.g., it would be non-linear. See also Ref. \cite{milegu} in this context.
  
\bibitem{krisnanda2017}
T.~Krisnanda, M.~Zuppardo, M.~Paternostro, and T.~Paterek,
\newblock Revealing nonclassicality of inaccessible objects,
\newblock {Phys. Rev. Lett.} {\bf119}, 120402 (2017).

\bibitem{groisman2005}
B.~Groisman, S.~Popescu, and A.~Winter,
\newblock Quantum, classical, and total amount of correlations in a quantum
  state,
\newblock {Phys. Rev. A} {\bf72}, 032317 (2005).

\bibitem{terhal2002}
B.~M. Terhal, M.~Horodecki, D.~W. Leung, and D.~P. DiVincenzo,
\newblock The entanglement of purification,
\newblock {J. Math. Phys.} {\bf43}, 4286 (2002).

\bibitem{horodecki2005}
M.~Horodecki, P.~Horodecki, R.~Horodecki, J.~Oppenheim, A.~Sen, U.~Sen, and
  B.~Synak-Radtke,
\newblock Local versus nonlocal information in quantum-information theory:
  Formalism and phenomena,
\newblock {Phys. Rev. A} {\bf71}, 062307 (2005).

\bibitem{vidal2002}
G.~Vidal and R.~F. Werner,
\newblock Computable measure of entanglement,
\newblock {Phys. Rev. A} {\bf65}, 032314 (2002).

\bibitem{piani2012}
M.~Piani,
\newblock Problem with geometric discord,
\newblock {Phys. Rev. A} {\bf86}, 034101 (2012).

\bibitem{paula2013}
F.~M. Paula, T.~R. Oliveira, and M.~S. Sarandy,
\newblock Geometric quantum discord through the Schatten 1-norm,
\newblock {Phys. Rev. A} {\bf87}, 064101 (2013).

\bibitem{messina2002}
A.~Messina,
\newblock A single atom-based generation of Bell states of two cavities,
\newblock {Eur. Phys. J. D} {\bf18}, 379 (2002).

\bibitem{browne2003}
D.~E. Browne and M.~B. Plenio,
\newblock Robust generation of entanglement between two cavities mediated by
  short interactions with an atom,
\newblock {Phys. Rev. A} {\bf67}, 012325 (2003).

\bibitem{devetak2005distillation}
I.~Devetak and A.~Winter,
\newblock Distillation of secret key and entanglement from quantum states,
\newblock {Proc. R. Soc. A} {\bf461}, 207 (2005).

\bibitem{horodecki2000limits}
M.~Horodecki, P.~Horodecki, and R.~Horodecki,
\newblock Limits for entanglement measures,
\newblock {Phys. Rev. Lett.} {\bf84}, 2014 (2000).

\bibitem{scully-book}
M.~O. Scully and M.~S. Zubairy,
\newblock {\em Quantum Optics}
\newblock (Cambridge University Press, Cambridge, UK, 1997).

\bibitem{dimwit1}
N. Brunner, S. Pironio, A. Acin, N. Gisin, A. A. Methot, and V. Scarani
\newblock Testing the dimension of Hilbert spaces
\newblock {Phys. Rev. Lett.} {\bf 100}, 210503 (2008).

\bibitem{dimwit2}
M. Hendrych, R. Gallego, M. Micuda, N. Brunner, A. Acin, and J. P. Torres
\newblock Experimental estimation of the dimension of classical and quantum systems
\newblock {Nat. Phys.} {\bf 8}, 588 (2012).

\bibitem{dimwit3}
J. Ahrens, P. Badziag, A. Cabello, and M. Bourennane
\newblock Experimental device-independent tests of classical and quantum dimensions
\newblock {Nat. Phys.} {\bf 8}, 592 (2012).


\bibitem{streltsov2012general}
A.~Streltsov, G.~Adesso, M.~Piani, and D.~Bru{\ss},
\newblock Are general quantum correlations monogamous?,
\newblock {Phys. Rev. Lett.} {\bf109}, 050503 (2012).

\bibitem{uhlmann1977relative}
A.~Uhlmann,
\newblock Relative entropy and the Wigner-Yanase-Dyson-Lieb concavity in an
  interpolation theory,
\newblock {Commun. Math. Phys.} {\bf54}, 21 (1977).

\bibitem{peres1996}
A. Peres,
\newblock Separability criterion for density matrices,
\newblock {Phys. Rev. Lett.} {\bf77}, 1413 (1996).

\bibitem{horodecki1996m}
M.~Horodecki, P.~Horodecki, and R. Horodecki,
\newblock Separability of mixed states: Necessary and sufficient conditions,
\newblock {Phys. Lett. A} {\bf223}, 1 (1996).

\bibitem{milegu}
X. Yuan, S. M. Assad, J. Thompson, J. Y. Haw, V. Vedral, T. C. Ralph, P. K. Lam, C. Weedbrook, and M. Gu, 
\newblock Replicating the benefits of Deutschian closed timelike curves without breaking causality, 
\newblock {npj Quantum Inf.} {\bf 1}, 15007 (2015).

\end{thebibliography}
\end{document}